\PassOptionsToPackage{svgnames,table}{xcolor}
\documentclass{sig-alternate-05-2015}

\usepackage[utf8]{inputenc}
\usepackage[english]{babel}

\usepackage{graphicx} % Required for including pictures
\usepackage{amsfonts}
\usepackage{amsmath}

\usepackage{float} % Allows putting an [H] in \begin{figure} to specify the exact location of the figure
\usepackage{amssymb}
\usepackage[ruled,vlined,linesnumbered]{algorithm2e}
	\SetKwInOut{Input}{input}
	\SetKwInOut{Output}{output}

\usepackage{subcaption}

\usepackage{paralist} % inline lists

\ifdefined \VersionLong

\else

\fi
\usepackage{tikz}
\usetikzlibrary{arrows,automata}
\tikzstyle{every node}=[initial text=]
\tikzstyle{location}=[rectangle, rounded corners, minimum size=12pt, draw=black, fill=blue!10, inner sep=2pt]
\tikzstyle{invariant}=[draw=black, dotted, inner sep=1pt] % xshift=1em, 
\usetikzlibrary{positioning}
\tikzset{main node/.style={circle,fill=blue!20,draw,minimum size=0.8cm,inner sep=0pt},
            }

\definecolor{darkblue}{rgb}{0, 0, 0.7}

\usepackage[
		pdfauthor={Author's name},%
		pdftitle={Document Title},
		breaklinks  = true,
		colorlinks  = true,
	\ifdefined \VersionWithComments
		pagebackref = true,
	\fi
	]{hyperref}

\usepackage[capitalise,english,nameinlink]{cleveref} % load after algorithm2e and hyperref
\crefname{line}{\text{line}}{\text{lines}} % to remove the capital

\newtheorem{proposition}{Proposition}
\newtheorem{theorem}{Theorem}
\newtheorem{definition}{Definition}
 \newtheorem{example}{Example}
 
\newtheorem{remark}{Remark}

\ifdefined\VersionWithComments
	\usepackage[colorinlistoftodos,textsize=footnotesize]{todonotes}
\else
	\usepackage[disable]{todonotes}
\fi

\ifdefined \VersionWithComments
	\newcommand{\todoinline}[1]{\mbox{}{\color{red}{\textbf{TODO}\ifx#1\\\else:\ \fi #1}}} % here, ``\\'' stands for ``empty''
\else
	\newcommand{\todoinline}[1]{}
\fi

\usepackage{verbatim} % for 'comment'

\usepackage[pagewise,switch]{lineno} % switch, modulo
\definecolor{cof}{RGB}{219,144,71}
\definecolor{pur}{RGB}{186,146,162}
\definecolor{greeo}{RGB}{91,173,69}
\definecolor{greet}{RGB}{52,111,72}
\definecolor{red}{RGB}{210,0,32}

\ifdefined\VersionWithComments
	\usepackage{draftwatermark}
	\SetWatermarkText{draft}
	\SetWatermarkScale{5}
	\SetWatermarkColor[gray]{0.9}
\fi
\ifdefined \VersionWithComments
 	\definecolor{colorok}{RGB}{80,80,150}
\else
	\definecolor{colorok}{RGB}{0,0,0}
\fi

\newcommand{\eg}{\textcolor{colorok}{e.\,g.,}\xspace}

\newcommand{\ie}{\textcolor{colorok}{i.\,e.,}\xspace}

\isbn{}
\doi{} 
\begin{document}
%\title{A Topological Method for Finding Invariant Sets of Switched Systems}
\title{Generation of bounded invariants  via 
stroboscopic set-valued maps: Application to the stability analysis
of parametric time-periodic systems
\todoinline{This is the version with comments. To disable comments, comment out line~3 in the \LaTeX{} source.}}
%method with central differences in space
%(forward in time, centered in space)
%explicit Euler three point finite difference scheme
%for dissipative reaction-diffusion equations}
\numberofauthors{2} %  in this sample file, there are a *total*
%\author{\ }
\author{
\alignauthor J. Jerray\\
\affaddr{Université Sorbonne Paris Nord, LIPN, CNRS, Villetaneuse, France} \\
\alignauthor L. Fribourg\\
\affaddr{Université Paris-Saclay, LSV, CNRS, ENS Paris-Saclay}
}
%

% For all page numbers, except p.1
\pagestyle{plain}

\maketitle

% For page numbers on p.1
\thispagestyle{plain}

\date{\today}
% Just remember to make, sure that the TOTAL number of authors
% is the number that will appear on the first page PLUS the
% number that will appear in the \additionalauthors section.

\maketitle

\todoinline{Instructions: 10 pages max (Excluding references);
    10pt font;
    two-column ACM format}

\begin{abstract}
A method is given for generating a bounded invariant of a differential system
with a given set of initial conditions around a point $x_0$. This
invariant has the form of a tube centered on the Euler approximate solution
starting at $x_0$,
which has for radius an upper bound on the distance between the approximate solution 
and the exact ones. The method consists in finding a real $T>0$
such that the ``snapshot'' of the tube at time $t=(i+1)T$
is included in the snapshot at $t=iT$, for some integer $i$. 
In the phase space, the invariant is therefore in the shape of a torus. 
A simple additional condition is also given to ensure that the solutions of the system can never converge to a point of equilibrium. In dimension 2, this ensures that all solutions converge towards a limit cycle. The method is extended in case the dynamic system contains a parameter $p$, thus allowing the stability analysis of the system for a range of values of $p$. This is illustrated on classical Van der Pol's system.

\end{abstract}

\category{G.1.7}{Mathematics of Computing}{Numerical Analysis}%{Ordinary Differential Equations}
\category{F.1.1}{Theory of Computation}{Computation by Abstract Devices}%{Models of Computation}
%\category{D.4.6}{Operating Systems}{Security and Protection}[Access Control]
%\category{K.6.5}{Management of Computing and Information Systems} {Security and Protection}
%A category including the fourth, optional field follows...
%\category{...}{...}{...}[...]

\terms{Algorithms, Theory, Verification}

%\keywords{ACM proceedings, \LaTeX, text tagging} % NOT required for Proceedings
\keywords{differential equations, periodicity, limit cycle, stability} % NOT required for Proceedings

%\RequirePackage{amsmath}
%\usepackage{amssymb}
%\usepackage{amsmath}
%\usepackage{mathabx}
% \usepackage{amsthm}
%\usepackage{graphicx}
%\usepackage{hyperref}
%\usepackage{color}
%\usepackage{algorithm,algpseudocode}
%\usepackage{todonotes}
% \usepackage{multicol,caption}
% \newenvironment{Figure}
%   {\par\medskip\noindent\minipage{\linewidth}}
%   {\endminipage\par\medskip}large time behavior solutions systems nonlinear reaction diffusion equations

%%%\graphicspath{{./figures/}}

% \and
% Co Author \qquad\qquad Yet S. Else
% \institute{Stanford Univeristy\\
% California, USA}
% \email{\quad is@gmail.com \quad\qquad somebody@else.org}
% }
%\titlerunning{Guaranteed convergence of
%explicit Euler three point finite difference scheme}

%\tableofcontents

\section{Introduction}

Given a differential system $\Sigma: dx/dt=f(x)$ of dimension $n$, an initial point $x_0\in\mathbb{R}^n$,
a real $\varepsilon>0$, and a ball $B_0=B(x_0,\varepsilon)$\footnote{$B(x_0,\varepsilon)$ is the set $\{z\in\mathbb{R}^n\ |\ \|z-x_0\|\leq\varepsilon\}$ where $\|\cdot\|$ denotes the Euclidean distance.},
%un point initial $x_0$, un reel $\varepsilon>0$ et un voisinage $B_0=B(x_0,\varepsilon)$ de $x_0$, 
we present here a simple method allowing to find a bounded invariant set
of $\Sigma$ containing the trajectories starting at $B_0$. This invariant set has the form of a tube
whose center at time~$t$ is the Euler approximate solution $\tilde{x}(t)$
of the system starting at $x_0$, and radius is a function $\delta_\varepsilon(t)$  bounding
the distance between $\tilde{x}(t)$ and an exact solution $x(t)$
starting at $B_0$. 
The tube can thus be described as $\bigcup_{t\geq 0} B(t)$
where $B(t)\equiv B(\tilde{x}(t), \delta_\varepsilon(t))$.

To find a {\em bounded} invariant, we then look for a positive real $T$
such that $B((i+1)T) \subseteq B(iT)$ for some $i\in\mathbb{N}$.
In case of success, the ball $B(iT)$  is guaranteed to 
contain the ``stroboscopic'' sequence $\{B(jT)\}_{j=i,i+1,\dots}$
of sets $B(t)$ at time $t=iT, (i+1)T, \dots$.
It follows that the bounded portion
$\bigcup_{t\in[iT,(i+1)T]} B(t)$  is equal to
$\bigcup_{t\in[iT,\infty)} B(t)$, 
and thus constitutes the sought bounded invariant set.

We apply the finding of such an invariant to the stability analysis of 
parameterized time-periodic differential systems.
We illustrate our results on the example of 
a parametric Van der Pol (VdP) system for which we show that, 
for a whole range of values of the parameter,
the solutions always converge towards a limit cycle.

%\cite{AminzareS12,AminzareS13,AminzareS14,Arcak11,SAAS13}

\paragraph{Comparison with related work}
We explain here some similarities and differences of our
method  with several kinds of related work.
\begin{itemize}
\item There exists a trend of work on the generation of torus-shaped invariants
using stroboscopic maps of quasi-periodic systems 
with possible parameters (see, e.g., \cite{Baresi-Scheeres,Gomez-Mondelo,Olikara-Scheeres}).
%
%COMPUTING INVARIANT TORI AND CIRCLES IN DYNAMICAL SYSTEMS VOLKER REICHELT
%
A first difference is that these works consider
stroboscopic mappings of {\em points} of ${\mathbb R}^n$, while our stroboscopic maps apply to
 {\em sets} of points.
A second difference is that they often use a
Fourier analysis in the {\em frequency domain} (using, e.g., the notion
of ``radii polynomials'' \cite{Castelli-Lessard}) while 
we remain in the {\em time domain}.

\item Our method makes use of a rigorous time-integration method in order to enclose the exact solutions with tubes,
which is similar  to what is done using high order
of Taylor method in ODE integration  as
proposed by Lohner or Taylor models \cite{Lohner87,Zgliczynski}.
Such methods are used in \cite{Castelli-Lessard,Kapela-Simo} 
%\cite{Kapela-Simo,Zgliczynski-Mischaikow} 
to rigorously compute
the eigenvalues of a so-called 
``monodromy matrix'', which allows to determine the linear stability of the
equilibrium points of the system.
%So the Interval  Krawczyk Method and the C1 –Lohner algorithm 
%(see, e.g., \cite{Zgliczynski,SchurmannA17b}).
%
Unlike these works,  our stability analysis does not
try to compute such eigenvalues of monodromy matrices.

\item Our method shares also some common features  with 
the works of \cite{Parrilo-Slotine,Aminzare-Sontag,vandenBerg-Queirolo}, which 
aim at proving a
{\em contractivity} property of the system
(i.e., that any two solutions converge exponentially to each other).
In~\cite{Aminzare-Sontag}, 
contractivity amounts  to the 
finding of a negative definitive quadratic form
(which is equivalent to the existence of a Lyapunov function for the system).
In~\cite{Parrilo-Slotine},
``Squares-of-Sum  programming is used to find ranges of uncertainty under which a system with uncertain perturbations is always contracting with the original contraction metric''.
In~\cite{vandenBerg-Queirolo}, they turn the stability problem
into the contractivity of a fixed point operator that is checked
with the assistance of a computer.
The difference here is that we do not try to prove a contractivity
property, but only the existence of two 
set-valued snapshots, one of which is included in the other one. 
This is a much weaker property and easier to prove.
\end{itemize}
Our method is simple and {\em a priori} efficient, but, as explained in
\cref{ex:vdp2}, cannot compete
with the specialized tools of the literature (which can use, e.g.,
a large number of Fourier modes) on complex
quasi-periodic systems like Van der Pol system with high values of parameter, 
or chaotic systems with strange attractors.

\paragraph{Plan of the paper}
In \cref{sec:method}, we present our method, then explain how to apply it to the stability analysis of parameterized systems in \cref{sec:ex}.
We conclude in \cref{sec:final}.

\section{Method}\label{sec:method}
\subsection{Euler's method and error bounds}\label{ss:Euler}
%, and 
%we suppose that the control law ${\bf u}(\cdot)$ is a {\em piecewise-constant} function, which takes its
%values on a {\em finite} set~$U$, called ``set of modes''.
%Given $u\in U$, l
Let us consider the differential system: 
%$$\frac{dy}{dt}=\sigma{\cal L}_h y+\sigma\varphi_h(t,u)+f(t,y)$$
%by:
$$\frac{dx(t)}{dt}=f(x(t)),$$
with states $x(t)\in\mathbb{R}^n$.
%
%$f_u(y(t))$ stands for $f({\bf u}(t),y(t))$ with ${\bf u}(t)=u$ for $t\in[0,\tau]$, and $y(r)\in\mathbb{R}^n$ denotes the state of the system at time $t$.
%
%{\cal L}_hy(t,x)+\varphi_h(t,u)+f_h(t,y).$$
%For $t\in [0,\tau]$ and $u\in U$,
We will  use $x(t;x_0)$ (or sometimes just $x(t)$) to denote  the exact continuous solution %~$x$ %:[0,\tau]\times \Omega_h\rightarrow [0,1]^n$ 
of the system at time~$t$,
for a given initial condition~$x_0$.
We use $\tilde{x}(t;y_0)$
(or just $\tilde{x}(t)$) to denote Euler's approximate value 
of $x(t;y_0)$
(defined by $\tilde{x}(t;y_0)=y_0+tf(y_0)$  for $t\in [0,\tau]$, where $\tau$ is the integration time-step).

We suppose that we know a bounded region
${\cal S}\subset \mathbb{R}^n$ containing
the solutions of the system for a set of initial conditions $B_0$
and  a certain amount of time.
%
%\subsection{Error bounds}\label{ss:error}
%We first consider the ODE:
%$\frac{dy}{dt}=f(y,{\bf 0})$, and 
We now give an upper bound to the error between
the exact solution of the ODE and
its Euler approximation on ${\cal S}$ (see~\cite{CoentF19,SNR17}).

\begin{definition}\label{def:delta}
Let $\varepsilon$ be a given positive constant. Let us define, for $t\in [0,\tau]$,
$\delta_{\varepsilon}(t)$ as follows:\\
%\begin{itemize}
%\item  
$\mbox{if } \lambda <0:$
$$\delta_{\varepsilon}(t)=\left(\varepsilon^2 e^{\lambda t}+
 \frac{C^2}{\lambda^2}\left(t^2+\frac{2 t}{\lambda}+\frac{2}{\lambda^2}\left(1- e^{\lambda t} \right)\right)\right)^{\frac{1}{2}}$$
%
%\item 
$\mbox{if } \lambda = 0:$
$$\delta_{\varepsilon}(t)= \left( \varepsilon^2 e^{t} + C^2 (- t^2 - 2t + 2 (e^t - 1)) \right)^{\frac{1}{2}}$$
%\item if $\lambda^a = 0:$
%$$\delta_t= \frac{C t^2}{2} + \delta$$
%
$\mbox{if } \lambda > 0:$
%\begin{dmath*}
$$\delta_{\varepsilon}(t)=\left(\varepsilon^2 e^{3\lambda t}+ 
\frac{C^2}{3\lambda^2}\left(-t^2-\frac{2t}{3\lambda}+\frac{2}{9\lambda^2}
\left(e^{3\lambda t}-1\right)\right)\right)^{\frac{1}{2}}$$
%\end{dmath*}
%
%\end{itemize}
where $C$ and $\lambda$ are real constants specific to function $f$,
defined as follows:
$$C=\sup_{y\in {\cal S}} L\|f(y)\|,$$
where $L$ denotes the Lipschitz constant for $f$, and
$\lambda$ is the ``one-sided Lipschitz constant'' (or ``logarithmic Lipschitz constant'' \cite{Aminzare-Sontag}) associated to $f$, \ie{} the 
minimal constant such that, for all $y_1,y_2\in {\cal S}$:
$$\langle f(y_1)-f(y_2), y_1-y_2\rangle \leq \lambda\|y_1-y_2\|^2,\ \ \ \ \ (H0)$$
where $\langle\cdot,\cdot\rangle$ denotes the scalar product of two vectors
of ${\cal S}$.
\end{definition}
The constant $\lambda$ can be computed using a nonlinear optimization solver (\eg{} CPLEX~\cite{cplex2009v12}) or using the Jacobian matrix of $f$ (see, \eg{}~\cite{Aminzare-Sontag}).

\begin{proposition}\label{prop:basic}\cite{SNR17}
Consider the solution $x(t;y_0)$ of $\frac{dx}{dt}=f(x)$ with
initial condition~$y_0$ and the approximate Euler solution
$\tilde{x}(t;x_0)$ with initial condition $x_0$.
For all $y_0\in B(x_0,\varepsilon)$, we have:
$$\|x(t;y_0)-\tilde{x}(t;x_0)\|\leq \delta_{\varepsilon}(t).$$
%where $\mu=\|y_0-z_0\|$.
\end{proposition}
\cref{prop:basic} underlies the principle of our set-based method
where set of points are represented as balls centered around the 
Euler approximate values of the solutions. This illustrated in~\cref{fig:illustration}: for any initial condition $x^0$ belonging
to the ball $B(\tilde{x}^0,\delta(0))$
with $\delta(0)=\varepsilon$,
the exact solution $x^1\equiv x(\tau;x^0)$ belongs to the ball $B(\tilde{x}^1,\delta_\varepsilon(\tau))$ where $\tilde{x}^1$ denotes the Euler approximation 
$\tilde{x}^0+\tau f(\tilde{x}^0)$ at 
$t=\tau$.
\begin{figure}[h!]
\centering
\includegraphics[scale=0.5]{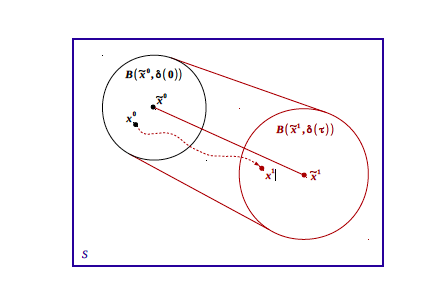}
\caption{Illustration of \cref{prop:basic}}
\label{fig:illustration}
\end{figure}
%We will suppose that the time step $\Delta t$ used in Euler's time-integration method is a fraction of the sampling period $\tau$ 
%of the switching system (e.g., $\Delta t=\tau/10$).
%Such a splitting is called ``sub-sampling'' in numerical methods (see, e.g.,
%\cite{???}.

\subsection{Systems with bounded uncertainty}

Let us now show how the method extends to 
systems with ``disturbance'' or ``bounded uncertainty''.
%\section{Robustness against bounded perturbations}\label{appendixun}
A differential system with bounded uncertainty is of the form 
$$\frac{dx(t)}{dt}=f(x(t),w(t)),$$
with $t\in\mathbb{R}^n_{\geq 0}$,
states $x(t)\in\mathbb{R}^n$, and uncertainty $w(t)\in{\cal W}\subset \mathbb{R}^n$ (${\cal W}$ is compact, \ie{} closed and bounded).
We assume that any possible disturbance trajectory is bounded at any
point in time in the compact set $W$. We denote this by
$w(\cdot)\in{\cal W}$, which is a shorthand for
$w(t)\in {\cal W},\forall t\geq 0$.
See~\cite{SchurmannA17,SchurmannA17b} for details. 
%
%\footnote{We denote this by $w(\cdot)\in{\cal W}$, which is a shorthand for $w(t)\in {\cal W}, \forall t\in [0,\tau]$.}. 
%The same shorthand is also used for state and input trajectories.} 
%\footnote{LF???:Without loss of understanding, we still use $x(t;y_0)$, or just $x(t)$, to denote the solution of $\frac{dx(t)}{dt}=f(x(t),w(t))$ with initial condition $y_0$. We also use $\tilde{x}(t; x_0)$, or just $\tilde{x}(t)$, to denote the approximate Euler solution  of the system $\frac{dx(t)}{dt}=f(x(t),0)$ {\em without uncertainty} (\ie{} when ${\cal W}=0$), with initial condition $x_0$.}
%
%The solution of an undisturbed system  (\ie{} with ${\cal W}=0$) is denoted by $\phi_{u}(t;y_0,0)$. 
%
%We now suppose that ${\cal S}$ is Euler-invariant for the
%system with {\em null perturbation}, \ie{} for all $y\in {\cal S}$,
%there exists $u$ such that $\tilde{Y}_{y,{\bf 0}}(\tau)\in {\cal S}$.
%\footnote{not restrictive??? compare to Assumption 5.1}
%and $\phi_\pi(\tau;y_0,0)\in {\cal S}$.
%
We now suppose (see~\cite{AdrienRP17}) that
there exist constants $\lambda\in\mathbb{R}$ 
and $\gamma\in\mathbb{R}_{\geq 0}$ such that,
for all $y_1,y_2\in {\cal S}$ and $w_1,w_2\in {\cal W}$:\\

$\langle f(y_1,w_1)-f(y_2,w_2), y_1-y_2\rangle$

$\leq 
\lambda\|y_1-y_2\|^2 + \gamma \|y_1-y_2\|\|w_1-w_2\|\ \ \ \ \ (H1).$\\
%where $\langle\cdot,\cdot\rangle$ denotes the scalar product of two vectors of ${\cal S}$.
\\
This formula can be seen as a generalization of (H0) (see \cref{ss:Euler}). 
Recall that $\lambda$ has to be computed in the absence of uncertainty
($|{\cal W}|=0$). The additional constant $\gamma$ is used for taking into account
the uncertainty $w$.
Given $\lambda$, the constant $\gamma$ can be computed itself
using a nonlinear optimization solver (\eg{} CPLEX~\cite{cplex2009v12}). Instead of computing them globally
for ${\cal S}$, it is advantageous to compute $\lambda$ and $\gamma$ {\em locally} depending on the subregion of ${\cal S}$ occupied by the system state during a considered interval of time.
Note that the notion of contraction (often used in the literature~\cite{ManchesterCDC13,Aminzare-Sontag}) corresponds to the case where $\lambda<0$, but we do not need this assumption here ($\lambda$ can be positive, at least locally). 
We now give a version of \cref{prop:basic} with bounded uncertainty $w(\cdot)\in{\cal W}$, originally proved in~\cite{AdrienRP17}.
\begin{proposition}\label{prop:1bis}\cite{AdrienRP17}
Consider
a system $\Sigma$ with bounded uncertainty
of the form $\frac{dx(t)}{dt}=f(x(t),w(t))$
satisfying (H1).
%and, for all $u\in U$:
%
%$e>gw/2  \wedge  e  < sqrt(2) C/l^2  \wedge e < 2/5^{1/4}  sqrt(C)/l^2)
%\wedge t*< sqrt(3)   e |l| / C$, 
%
%where $t^*$ is the positive real
%root of equation
%$F(t)= Pt^2 + Qt +R = 0$
%		with $ P = –(d|l|^3/6), Q = (a+dl^2/2), R = (b-d|l|)$.

Consider a point~$x_0\in {\cal S}$ and a point $y_0\in B(x_0,\varepsilon)$.
Let $x(t;y_0)$ be the exact solution of the system 
$\frac{dx(t)}{dt}=f(x(t),w(t))$ with bounded uncertainty ${\cal W}$
and initial condition $y_0$,
and $\tilde{x}(t;x_0)$ the Euler approximate solution of the system
$\frac{dx(t)}{dt}=f(x(t),0)$ {\em without uncertainty}
($|{\cal W}|=0$) with initial condition $x_0$.
%of $\varepsilon$-representative $z_0\in {\cal X}$.
% Let us denote by $\lambda_{u}$ the greatest eigenvalue 
% of $\frac{A_{u} + A_{u}^\top}{2}$ for all $u \in U_1$.
% Suppose sub-system 1 verifies $$f_{\sigma_1} (x_1,x_2) = A x_1 + B x_2 + u_{\sigma_1}$$
We have,
for all $w(\cdot) \in {\cal W}$ and $t\in[0,\tau]$: %and any $\sigma_2 \in U_2$:

%$$\phi_j(t;y_0)\in {\cal B}(\tilde{y}_0-t f_j(\tilde{y}_0), \gamma)$$
$$\|x(t;y_0)-\tilde{x}(t;x_0)\|\leq \delta_{\varepsilon,{\cal W}}(t).$$
with %$y_0 = \begin{pmatrix}y_0 \\y_2_0\end{pmatrix}$ and
\begin{itemize}
\item if $\lambda <0$,
\begin{multline}
 \delta_{\varepsilon,{\cal W}}(t) = 
%  \frac{1}{\lambda_{u}^{3/2}} 
 \left( \frac{C^2}{-\lambda^4} \left( - \lambda^2 t^2 - 2 \lambda t + 2 e^{\lambda t} - 2 \right) \right.   \\
 + \left. \frac{1}{\lambda^2} \left( \frac{C \gamma |{\cal W}|}{-\lambda} \left( - \lambda t + e^{\lambda t} -1 \right) \right. \right.  \\ + \left. \left. \lambda \left( \frac{\gamma^2 (|{\cal W} |/2)^2}{-\lambda} (e^{\lambda t } - 1) + \lambda \varepsilon^2 e^{\lambda t}  \right) \right)  \right)^{1/2}
\end{multline}
%\begin{itemize}
% \item if $\lambda^1_{j_1} <0$,
%\begin{multline}
% \delta_{j_1} (t) = 
%  \frac{1}{\lambda^1_{j_1}^{3/2}} 
% \left( \frac{(C_{j_1}^1)^2}{-(\lambda^1_{j_1})^4} \left( - (\lambda^1_{j_1})^2 t^2 - 2 \lambda^1_{j_1} t + 2 e^{\lambda^1_{j_1} t} - 2 \right) \right.   \\
% + \left. \frac{1}{(\lambda^1_{j_1})^2} \left( \frac{C_{j_1}^1 \gamma^1_{j_1} |S_2|}{-\lambda^1_{j_1}} \left( - \lambda^1_{j_1} t + e^{\lambda^1_{j_1} t} -1 \right) \right. \right.  \\ + \left. \left. \lambda^1_{j_1} \left( \frac{(\gamma^1_{j_1} )^2 (|S_2 |/2)^2}{-\lambda^1_{j_1}} ( e^{\lambda^1_{j_1} t } - 1) + \lambda^1_{j_1} \delta^2 e^{\lambda^1_{j_1} t}  \right) \right)  \right)^{1/2}
%\end{multline}
\item if $\lambda >0$,
\begin{multline}
 \delta_{\varepsilon,{\cal W}}(t) = \frac{1}{(3\lambda)^{3/2}} \left( \frac{C^2}{\lambda} \left( - 9\lambda^2 t^2 - 6\lambda t + 2 e^{3\lambda t} - 2 \right) \right.   \\
 + \left. 3\lambda \left( \frac{C  \gamma |{\cal W}|}{\lambda} \left( - 3\lambda t + e^{3\lambda t} -1 \right) \right. \right.  \\
 + \left. \left. 3\lambda \left( \frac{\gamma^2 (|{\cal W} |/2)^2}{\lambda} ( e^{3\lambda t } - 1) + 3\lambda \varepsilon^2 e^{3\lambda t}  \right) \right)  \right)^{1/2}
\end{multline}
\item if $\lambda = 0$, 
\begin{multline}
 \delta_{\varepsilon,{\cal W}}(t)= 
%  \frac{1}{\lambda^{3/2}} 
 \left( {C^2} \left( -  t^2 - 2  t + 2 e^{ t} - 2 \right) \right.   \\
 + \left.  \left( {C \gamma |{\cal W}|} \left( -  t + e^{ t} -1 \right) \right. \right.  \\ + \left. \left.  \left({\gamma^2 (|{\cal W} |/2)^2} ( e^{ t } - 1) +  \varepsilon^2 e^{ t}  \right) \right)  \right)^{1/2}
\end{multline}
\end{itemize}
\end{proposition}
We will sometimes write $\delta_{{\cal W}}(t)$ instead of
$\delta_{\varepsilon,{\cal W}}(t)$.
%Let $\delta_{{\cal W}}(t)\equiv  \delta_{\varepsilon,{\cal W}}(t)$
%Let $B(t)\equiv B(\tilde{Y}_{x_0,{\bf 0}}(t),\delta_{\varepsilon,{\cal W}}(t))$.
%with $\tilde{y}_{{\cal W}}(0)=z_0$.
%\cref{prop:1bis} expresses that, for $t\in[0,\tau]$, the ``tube'' 
%$\bigcup_{t\geq 0}B(t)$ contains all the 
%solutions $Y_{y_0,{\cal W}}(t)$ with $\|y_0-x_0\|\leq \varepsilon$,
%and is therefore {\em robustely (positive) invariant}.
\subsection{Correctness}
Consider a differential system $\Sigma: dx/dt=f(x,w)$
with $w\in{\cal W}$, an initial point $x_0\in\mathbb{R}^n$, a real $\varepsilon>0$ and a ball $B_0=B(x_0,\varepsilon)$. Let
$B(t)$ denote $B(\tilde{x}(t), \delta_{\varepsilon,{\cal W}}(t))$
where $\tilde{x}(t)$ is the Euler approximate solution of the system
without uncertainty and initial condition $x_0$\footnote{Note that $B(0)=B_0$ because
$\tilde{x}(0)=x_0$ and $\delta_{\varepsilon,{\cal W}}(0)=\varepsilon$.}.
It follows from \cref{prop:1bis} that $\bigcup_{t\geq 0} B(t)$ 
is an invariant set containing $B_0$.
We can make an a stroboscopic map of this invariant.
by considering periodically the set $B(t)$ at the moments
$t=0,T,2T$, etc., with $T=k\tau$ for some $k$ 
($\tau$ is the time-step used un Euler's method).

If moreover, we can find an integer $i\geq 0$ such that 
$B((i+1)T)\subseteq B(iT)$, then
%la suite $\{B(jT)\}_{j=i,i+1,\dots}$ est decroissante pour la relation d'ordre 
%d'inclusion des ensembles, et 
we have
$B(iT)=\bigcup_{j=i,i+1,\dots}B(jT)$
%$\bigcup_{t\geq 0}B(t)=\bigcup_{t\in[0,(i+1)T]}B(t)$ et
and $\bigcup_{t\in[iT,(i+1)T]}B(t)=\bigcup_{t\in[iT,\infty)}B(t)$.
The set $\bigcup_{t\in[iT,(i+1)T]}B(t)$ (abbreviated to $B([iT,(i+1)T])$)
is thus a {\em bounded invariant} 
which contains all the solutions $x(t)$
starting at $B_0$, for $t\in[iT,\infty)$.
In the phase space, this bounded invariant has a ``torus'' shape.
We have:
\begin{proposition}\label{prop:inclusion}
Suppose that there exist $T>0$ (with $T=k\tau$ for some $k$) 
and $i\in\mathbb{N}$ such that: $B((i+1)T) \subseteq B(iT)$. Then
%
%\begin{enumerate}
%\item $B(iT)=\bigcup_{j=i,i+1,\dots}B(jT)$,  
%and $\bigcup_{0\leq t\leq (i+1)T } B(t)$ is a compact 
%invariant set containing $B_0$.
%\item 
$B[iT,(i+1)T]\equiv \bigcup_{t\in[iT,(i+1)T]} B(t)$ is a
compact (i.e., bounded and closed) invariant set containing, for $t\in[iT,\infty)$, all the solutions $x(t)$ of $\Sigma$  with initial condition in~$B_0$.
%, and bounded uncertainty $w\in {\cal W}$.
%\end{enumerate}
\end{proposition}
\begin{remark} %(a mettre plus haut???)
The value $T$ is an approximate value of the exact period $T^*$ of the system.
Actually, for a system involving a  parameter $p$ (see \cref{sec:ex}), the exact value of the period of the system depends on the 
value of $p$. However, we will see (see \cref{sec:ex}) that a same approximate value
$T$ allows us to analyze the system for a range of values of the 
parameter $p$, and thus a whole range of periods of the system.
\end{remark}

\begin{remark}
Note that the limit ball $\bigcap_{j=i,i+1,\dots}B(jT)$ does not usually have a zero diameter (whereas, under the assumption of point-to-point contraction, it always does). So we cannot use our method to immediately show the convergence to a limit cycle (see, e.g., \cite{Parrilo-Slotine}) except in the 2-dimensional case where
Poincar\'e-Bendixson's theorem applies (see \cref{ex:vdp1}). In the general case
($n>2$), our method only allows us to show that the system remains indefinitely
in a spatially bounded region, and under an additional condition,
never converges towards an equilibrium point (see \cref{sec:ex}).
\end{remark}
%\begin{remark}
%side effect: l'inclusion $B(i+1)\subseteq B(i)$ est souvent
%plus facile a montrer numeriquement 
%(et visuellement plus nette sur une figure) dans une version avec uncertainty
%que sans uncertainty (${\cal W}=0$).
%\end{remark}
%\section{Application to Parametric Analysis of Stability}
\section{Application to the Stability Analysis of Parametric Systems}\label{sec:ex}
Let us now consider a family $\{\Sigma_p\}_{p\in {\cal P}}$
of differential systems $\Sigma_p$ of the form
$dx/dt=f_p(x)$ involving a {\em parameter} $p\in {\cal P}$
(but no uncertainty).
It is useful to find a 
subset ${\cal Q}$ of ${\cal P}$  
and a system $\Sigma':dx/dt=f(x,w)$ with uncertainty such that, for any $p\in{\cal Q}$,
$\Sigma_p$ is a particular form of $\Sigma'$  
for an appropriate uncertainty function $w(\cdot)\in {\cal W}$.
%and an initial set of conditions $B(x_0,\varepsilon)$.
%
%This allows to consider the solution of $\Sigma_p:dx/dt=f_p(x)$ 
%for $p\in{\cal Q}$,
%as a {\em particular} solution of 
%$\Sigma':dx/dt=f(x,w)$ with $w\in{\cal W}$.
%
This is useful to infer certain common properties of the solutions
of $\{\Sigma_p\}_{p\in{\cal Q}}$  from the analysis of the system 
$\Sigma'$ with uncertainty (cf. \cite{Parrilo-Slotine}, Section~5).

\begin{example}\label{ex:vdp0}
Consider the Van der Pol (VdP) system $\Sigma_p$ of dimension $n=2$ with parameter $p\in\mathbb{R}$,
and initial condition in $B_0=B(x_0,\varepsilon)$
for some $x_0\in\mathbb{R}^2$ and $\varepsilon>0$ (see~\cite{vandenBerg-Queirolo}):

$du_1/dt=u_2$

$du_2/dt=p u_2 -p u_1^2u_2-u_1$.\\
Consider now  the system $\Sigma'$ with uncertainty $w(\cdot)\in{\cal W}_0=[-0.5,0.5]$ and initial condition $x_0$:

$du_1/dt=u_2$

$du_2/dt=(p_0+w) u_2 -(p_0+w) u_1^2u_2-u_1$,\\
with $p_0=1.1$.
It is easy to see that
each solution of $\Sigma_p$ with $p\in[p_0-0.5,p_0+0.5]=[0.6,1.6]$
is a particular solution
of system $\Sigma'$.
\end{example}

\begin{proposition}\label{prop:noequil}
Suppose that, for some index $1\leq j\leq n$, we have $m_+^j<M_-^j$ where
$m_+^j$ (resp. $M_-^j$) denotes the minimum 
(resp. maximum) of $\tilde{x}^j(t)+\delta_{\varepsilon,{\cal W}}(t)$
(resp. $\tilde{x}^j(t)-\delta_{\varepsilon,{\cal W}}(t)$)
for $t\in[iT,(i+1)T]$.
Then $B[iT,(i+1)T]$ contains no fixed point
of $\Sigma'$.
\end{proposition}
\begin{proof}
By {\em reductio ad absurdum.}
Suppose that there is an uncertainty function $w(\cdot)\in{\cal W}$
such that there exists a fixed point $y_0$ of $\Sigma'$ in
$B[iT,(i+1)T]\equiv \bigcup_{t\in[iT,(i+1)T]}B(\tilde{x}(t),\delta_{{\cal W}}(t))$. Then the $j$-coordinate $y_0^j$ of $y_0$ satisfies
$\tilde{x}^j(t)-\delta_{{\cal W}}(t)\leq y^j\leq  \tilde{x}^j(t)+\delta_{{\cal W}}(t)$ for all $t\in [iT,(i+1)T]$. It follows that
$M_-^j\leq y_j\leq m_+^j$, which contradicts the hypothesis $m_+^j<M_-^j$.
\end{proof}

\begin{theorem}\label{th:inclusion}
Given a system $dx/dt=f(x,w)$ with uncertainty $w\in{\cal W}$ and
a set of initial conditions $B_0\equiv B(x_0,\varepsilon)$, suppose:
\begin{enumerate}
\item any solution $x(t)$ of the parametric system 
$\Sigma_p: dx/dt=f_p(x)$ with $p\in {\cal Q}\subseteq {\cal P}$ 
and initial condition $x_0$ is a solution of the  system 
$\Sigma': dx/dt=f(x,w)$ with initial set of conditions $B_0$
for some uncertainty function $w(\cdot)\in{\cal W}$.
\item $\Sigma'$ is such that there exist $T>0$ (with $T=k\tau$ for some $k$) 
and $i\in\mathbb{N}$: $B((i+1)T) \subseteq B(iT)$.
\item $m_+^j<M_-^j$ for some $1\leq j\leq n$ (where $m_+^j$ and $M_-^j$
are defined as in \cref{prop:noequil}).
%\item there is no $y_0\in B[iT,(i+1)T]$ such that $f(y_0)=0$
%(i.e., $B[iT,(i+1)T]$ contains no equilibrium point).
%such that 
%
%\ \ \ $H(y_0)\subseteq ,\\
%where $H(y_0)$ is the ``horizontal segment'' defined by:
%\ \ \ $H(y_0)\equiv \{y(t)\ |\ \ y(t)=y_0,\ t\in[iT,(i+1)T]\}$.
\end{enumerate}
%B(\tilde{x}(iT),
Then  no solution $x(t)$ of $\Sigma_p$ with $p\in {\cal Q}$, 
converges towards a point of $\mathbb{R}^n$ when $t\rightarrow\infty$.
In the case of dimension $n=2$, this implies that every solution
$x(t)$ of $\Sigma_p$ with $p\in{\cal Q}$ converges towards a limit cycle.
%by checking that \delta_\varepsilon(iT)\}_{t\in[iT,(i+1)T}$.
%$B(\tilde{x}(iT),\delta_\varepsilon(iT)$ and $B(\tilde{x}((i+1)T),\delta_\varepsilon((i+1)T)$.
\end{theorem}
\begin{proof} By {\em reductio ad absurdum}.
Suppose that $y_0\in \mathbb{R}^n$ is the limit
of a solution $x(t)$ of $\Sigma_p$ with $p\in{\cal Q}$ when $t\rightarrow\infty$.
Since $x(t)$ is a solution of $\Sigma'$
and $B[(i+1)T]\subseteq B[iT]$ by hypotheses 1 and 2 of \cref{th:inclusion}, then $x(t)\in B[iT,(i+1)T]$ for all $t>0$ by \cref{prop:inclusion}.
So the limit $y_0$ belongs to the adherence of $B[iT,(i+1)T]$, which is equal
to $B[iT,(i+1)T]$ (since $B[iT,(i+1)T]$ is closed).
On the other hand, by hypothesis~3 of \cref{th:inclusion} and \cref{prop:noequil}, 
$B[iT,(i+1)T]$ contains no fixed point, whence a contradiction.
\end{proof}

The implementation has been done in Python
% 	by Jawher Jerray.
and corresponds to a program of
around 500 lines.
The source code is available at
	\href{https://lipn.univ-paris13.fr/~jerray/parameter/}{\nolinkurl{lipn.univ-paris13.fr/~jerray/parameter/}}.
In the experiments below, the program runs on a 2.80 GHz Intel Core i7-4810MQ CPU with 8\,GiB of memory.
Given $x_0$, one searches for values of $\tau,\varepsilon,p,|{\cal W}|,T$ at hand
(by trial and error)
in order to make the program verify $B(iT)\subset B((i+1)T)$ for some $i\in\mathbb{N}$, and $m_+^j<M_-^j$ for some $1\leq j\leq n$. This is illustrated
in the following examples.

%On peut refaire la demarche pour toute autre valeur du point $x_0$. 
%L'extension de la \cref{prop:inclusion} au cas avec uncertainty
%est immediate.

\begin{example}\label{ex:vdp1}
Consider the system $\Sigma_p$ of \cref{ex:vdp0} 
and the system $\Sigma'$ with uncertainty $|{\cal W}_0|= 0.5$,
and let $p_0=1.1$.
%
%$du_1/dt=u_2$
%
%$du_2/dt=p u_2 -p u_1^2u_2-u_1$
%
%$du_1/dt=u_2$
%
%$du_2/dt=(p_0+w_0) u_2 -(p_0+w_0) u_1^2u_2-u_1$
%
%o\`u $p_0=1.1$ et 
%et $w_0$ represente une incertitude (scalaire) d'amplitude $\leq 0.5$.
%
%écart initial $\varepsilon_0=0.5$, $x_0=(u_1(0),u_2(0))=(1.7018, -0.1284)$, $T_0=6.746$.
%
Each solution of $\Sigma_p$ with $p\in[p_0-0.5,p_0+0.5]=[0.6,1.6]$
and initial condition $x_0$ is thus a particular solution of the system $\Sigma'$ with uncertainty
$|{\cal W}_0|= 0.5$ and set of initial conditions $B(x_0,\varepsilon_0)$ with $\varepsilon_0=0.2$.
For $x_0=(u_1(0),u_2(0))=(1.70177925, -0.12841500)$, 
$\tau=1/1000$,
$T_0=6.746=6746\tau$,
we find:

$\tilde{x}(0)   = (1.70177925, -0.12841500),  \delta_{{\cal W}}(0)  = 0.2$

$\tilde{x}(T_0)   = (1.86291588, -0.49830119)$, 

\hspace*{\fill} $\delta_{{\cal W}}(T_0)  = 1.80979852$

$\tilde{x}(2T_0) = (1.86246936, -0.49899153)$,  

\hspace*{\fill} $\delta_{{\cal W}}(2T_0) = 1.81704178$

$\tilde{x}(3T_0) = (1.86195743, -0.49951212)$,  

\hspace*{\fill} $\delta_{{\cal W}}(3T_0) = 1.83409614$

$\tilde{x}(4T_0) = (1.86144495, -0.50003175)$,  

\hspace*{\fill} $\delta_{{\cal W}}(4T_0) = 1.82669790$

$\tilde{x}(5T_0) = (1.86093193, -0.50055048)$,  

\hspace*{\fill} $\delta_{{\cal W}}(5T_0) = 1.81964175$

It follows: $B((i_0+1)T_0)\subset B(i_0T_0)$ for $i_0=3$.

On \cref{fig:illustration3}, the curve of Euler's approximation
$\tilde{x}(t)$ is depicted in red for initial condition $x_0$
and $t\in[0,5T_0]$. The borders of the tube 
$B(t)\equiv B(\tilde{x}(t),\delta_{\varepsilon,{\cal W}}(t))$
are depicted in green.
Various simulation curves are given in blue, and one can check
that they lie inside the green borders of the tube
(by \cref{prop:1bis}). The black vertical lines delimit the portion of the tube
between $t=i_0T_0$ and $t=(i_0+1)T_0$ (i.e., $B[i_0T_0,(i_0+1)T_0]$).
We also see on \cref{fig:illustration3} that
the minimum $m^1_+$ (represented by a small cyan ball)
of the upper green curve $\tilde{u}_1(t)+\delta_{{\cal W}}(t)$
 is less than the maximum $M^1_-$ (represented by a small gray ball)
of the lower green curve $\tilde{u}_1(t)-\delta_{{\cal W}}(t)$
(numerically, $m^1_+=-0.55695< M^1_-=-0.55411$).
It follows by \cref{th:inclusion} that no
solution $x(t)$ from $B(x_0,\varepsilon_0)$ converges to a point of $\mathbb{R}^n$.
As explained in \cref{ex:vdp0}, this shows that whatever the value of 
$p\in[p_0-|{\cal W}_0|,p_0+|{\cal W}_0|]=[0.6,1.6]$,
the solution of $\Sigma_p$ never converges to
a point of $\mathbb{R}^n$.
Since the size of the system is $n=2$, it follows by Poincar\'e-Bendixson's theorem that the solution of $\Sigma_p$ converges 
always towards a limit circle
for any $p\in[0.6,1.6]$ and initial condition in $B(x_0,\varepsilon_0)$.

The computation took $31s$.

%\footnote{Objectif: find the minimal number of Fourier modes K necessary for the validation of a periodic orbit  depending on $p$. The system has a unique attracting limit cycle for any $p>0$. In the limit $p\rightarrow 0$ it tends to a circle in phase space and is essentially described by a single Fourier mode. For large $p$ its dynamics fall in the fast-slow paradigm and more and more Fourier modes are required to describe the orbit accurately. We first fix $p=1.1$. We then fix $p=2$ and vary $p$  to see how the minimal number of Fourier modes needed to describe a validated solution increases as the complexity of the solution increases. The smallest $p$ chosen here is $p= 10^{-4}$. For $p$ even closer to 0 ($p< 10^{-8}$), we cannot validate the solution, since it loses its isolation property as it merges into a one  parameter family for $p=0$.}

\begin{figure}[h!]
\centering
\includegraphics[scale=0.4]{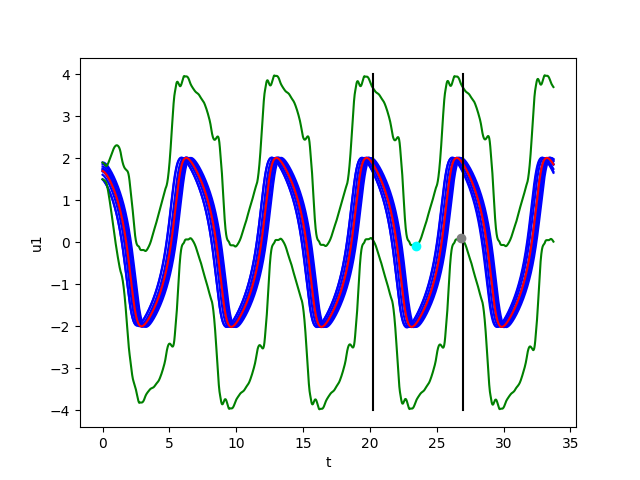}
\includegraphics[scale=0.4]{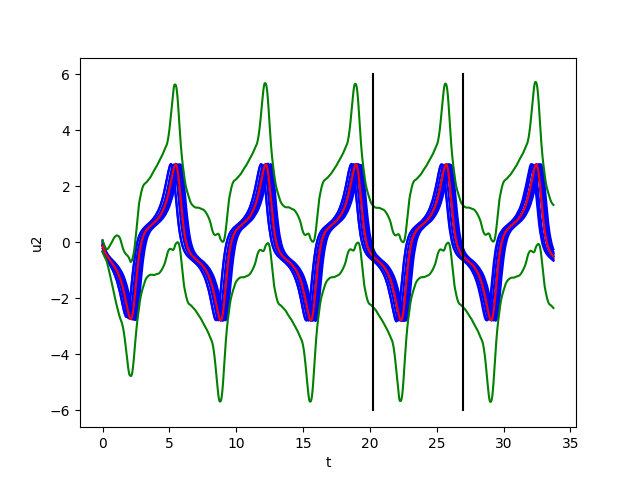}
\caption{VdP system
with parameter $p_0 = 1.1$, uncertainty $|{\cal W}_0|= 0.5$, initial radius $\varepsilon_0=0.2$, initial point $x_0=(1.7018, -0.1284)$, period $T_0=6.746$,
time-step $\tau=10^{-3}$.}
\label{fig:illustration3}
\end{figure}
\end{example}

\begin{example}\label{ex:vdp2}
We repeat the same experience as in \cref{ex:vdp1}, 
but for the value of the parameter $p_1 = 0.4$, the amplitude
of uncertainty $|{\cal W}_1|= 0.2$.
Each solution of $\Sigma_p$ with $p\in[p_1-0.2,p_1+0.2]=[0.2,0.6]$
and initial condition $x_0$ is 
now a particular solution of $\Sigma'$ with uncertainty
$|{\cal W}_1|= 0.2$ and set of initial conditions $B(x_0,\varepsilon_1)$ with $\varepsilon_1=0.2$.
For $T_1 = 6.347$, we find:

 $\tilde{x}(0)  = (1.70177925, -0.12841500),  \delta_{{\cal W}}(0)  = 0.2$

$\tilde{x}(T_1)  = (1.93787992, -0.35063219)$,  

\hspace*{\fill} $\delta_{{\cal W}}(T_1)  = 1.15449575$
    
$\tilde{x}(2T_1) = (1.96045193, -0.37539412)$,  

\hspace*{\fill} $\delta_{{\cal W}}(2T_1) = 1.40421766$
    
$\tilde{x}(3T_1) = (1.96234271, -0.37680622)$,  

\hspace*{\fill} $\delta_{{\cal W}}(3T_1) = 1.43088067$
    
$\tilde{x}(4T_1) = (1.96262148, -0.37636672)$,  

\hspace*{\fill} $\delta_{{\cal W}}(4T_1) = 1.42796887$
    
$\tilde{x}(5T_1) = (1.96277676, -0.37578520)$,  

\hspace*{\fill} $\delta_{{\cal W}}(5T_1) = 1.42671414$.\\
We have: $B((i_1+1)T_1)\subset B(i_1T_1)$ for $i_1=3$.
We now check on \cref{fig:illustration4}:  $m^1_+<M^1_-$,
It follows by \cref{th:inclusion}  that no
solution $x(t)$ from $B(x_0,\varepsilon_1)$ converges to a point of $\mathbb{R}^n$. This shows that whatever the value of 
$p\in[p_1-|{\cal W}_1|,p_1+|{\cal W}_1|]=[0.2,0.6]$,
the solution of $\Sigma_p$ never converges to
a point of $\mathbb{R}^n$.
It follows by Poincar\'e-Bendixson's theorem that the solution of $\Sigma_p$ converges 
always towards a limit circle
for any $p\in[0.2,0.6]$ and initial condition in $B(x_0,\varepsilon_1)$.

The computation time is $31s$.

\begin{figure}[h!]
\centering
\includegraphics[scale=0.4]{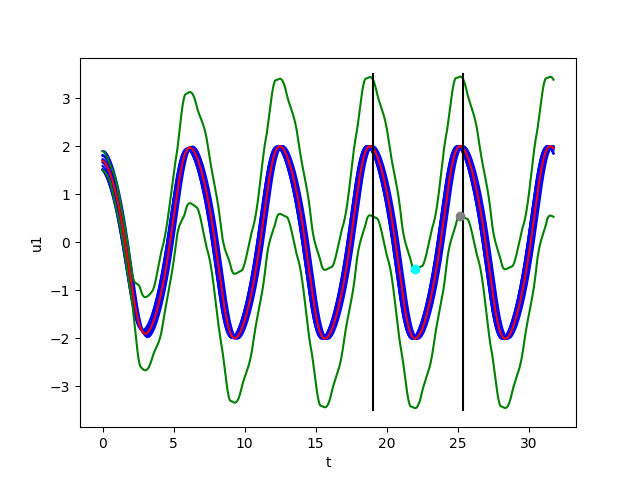}
\includegraphics[scale=0.4]{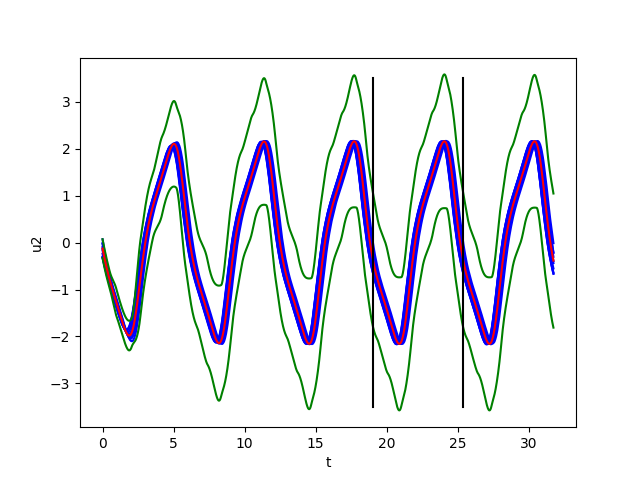}
\caption{VdP system
with parameter $p_1 = 0.4$, uncertainty $|{\cal W}_1|= 0.2$, initial radius $\varepsilon_1=0.2$, initial point $x_0=(1.7018, -0.1284)$, 
period $T_1=6.347$,
time-step $\tau=10^{-3}$.}
\label{fig:illustration4}
\end{figure}
\end{example}

\begin{example}\label{ex:vdp3}
We repeat again the same experience as in \cref{ex:vdp1}, 
but for the value of the parameter $p_2 = 1.9$, the amplitude
of uncertainty $|{\cal W}_2|= 0.3$, and initial radius $\varepsilon_2=0.1$.
Each solution of $\Sigma_p$ with $p\in[p_2-0.3,p_2+0.3]=[1.6,2.2]$
and initial condition $x_0$ is 
now a particular solution of $\Sigma'$ with uncertainty
$|{\cal W}_2|= 0.3$ and set of initial conditions $B(x_0,\varepsilon_2)$.
For $T_2 = 7.531$, we find:

    $\tilde{x}(0)  = (1.70177925, -0.12841500),  \delta_{{\cal W}}(0)  = 0.1$

    $\tilde{x}(T_2)  = (1.77514170, -0.39762673)$,  

\hspace*{\fill} $\delta_{{\cal W}}(T_2)  = 2.26229880$

    $\tilde{x}(2T_2) = (1.77450915, -0.39786534)$,  

\hspace*{\fill} $\delta_{{\cal W}}(2T_2) = 2.29620538$
    
$\tilde{x}(3T_2) = (1.7738762,  -0.39810409)$,  

\hspace*{\fill} $\delta_{{\cal W}}(3T_2) = 2.33196104$

$\tilde{x}(4T_2) = (1.77324288, -0.39834298)$,  

\hspace*{\fill} $\delta_{{\cal W}}(4T_2) = 2.25291437$

$\tilde{x}(5T_2) = (1.77260918, -0.39858202)$,  

\hspace*{\fill} $\delta_{{\cal W}}(5T_2) = 2.18787865$.

We have here: $B((i_2+1)T_2)\subset B(i_2T_2)$ for $i_2=3$.
We now check on \cref{fig:illustration5}:  $m^1_+<M^1_-$,
It follows by \cref{th:inclusion}  that no
solution $x(t)$ from $B(x_0,\varepsilon_2)$ converges to a point of $\mathbb{R}^n$. This shows that whatever the value of 
$p\in[p_2-|{\cal W}_2|,p_2+|{\cal W}_2|]=[1.6,2.2]$,
the solution of $\Sigma_p$ never converges to
a point of $\mathbb{R}^n$.
It follows by Poincar\'e-Bendixson's theorem that the solution of $\Sigma_p$ converges 
always towards a limit circle
for any $p\in[1.6,2.2]$ and initial condition in $B(x_0,\varepsilon_2)$.
The computation took $31s$.
Using the results of \cref{ex:vdp1} and \cref{ex:vdp2}, 
we know actually that, for any $p\in[0.2,2.2]$, the solution of $\Sigma_p$ always converges towards a limit cycle, for initial condition in $B(x_0,\varepsilon_2)$.

For each \cref{ex:vdp1,ex:vdp2,ex:vdp3}, once the inputs $\tau,\varepsilon,|{\cal W}|,T$ have 
been fixed, the program takes less than 100s to find $i$ such that
$B(T(i+1))\subset T(i)$ and $m_+< M_-$. The method is thus efficient
for treating such a range of parameters.
However, for $p>2.2$, it becomes difficult to find appropriate values 
of $T$ and $|{\cal W}|$ 
that make $B(T(i+1))\subset T(i)$ provable for some $i\in\mathbb{N}$, and
this points to a current limit of the method. 
As explained in 
\cite{vandenBerg-Queirolo}:
``The VdP system has a unique attracting limit cycle for
any $p>0$. In the limit $p\downarrow 0$ it tends to a circle in phase space and is essentially
described by a single Fourier mode. For large $p$ its dynamics fall in the fast-slow
paradigm and more and more Fourier modes are required to describe the orbit
accurately''. In \cite{vandenBerg-Queirolo},
they manage to prove the convergence towards a limit cycle for $p=9$, using 
more than 1000 Fourier modes (the computation time is cubic in the
number of modes). Note on the other hand, that their approach, unlike ours, is limited to systems with {\em polynomial} vector field~$f$.

\begin{figure}[h!]
\centering
\includegraphics[scale=0.4]{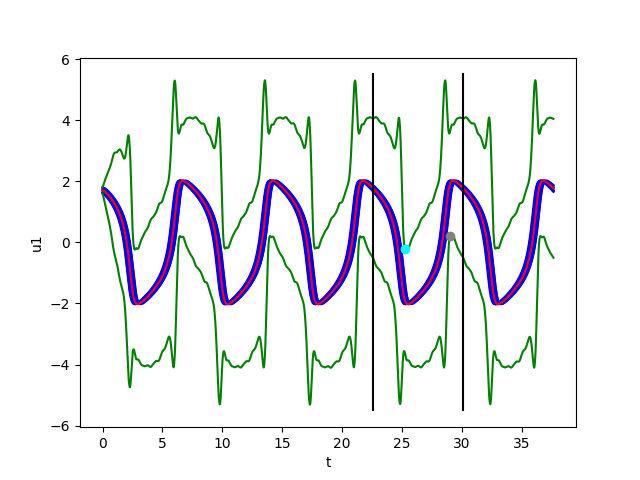}
\includegraphics[scale=0.4]{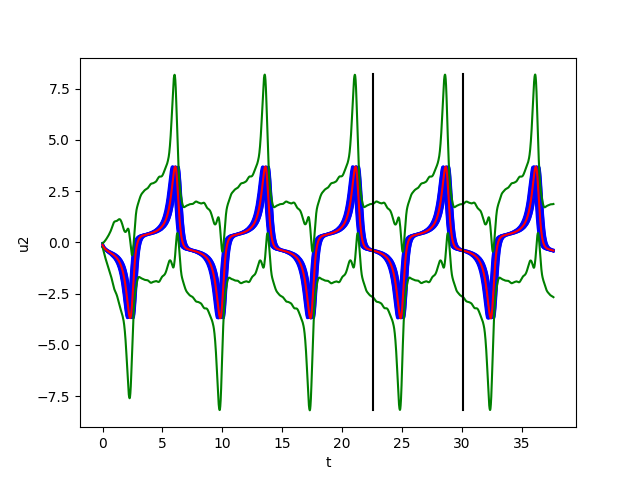}
\caption{VdP system
with parameter $p_2 = 1.9$, uncertainty $|{\cal W}_2|=0.3$, initial radius $\varepsilon_2=0.1$, initial point $x_0=(1.7018, -0.1284)$, 
period $T_2=7.531$,
time-step $\tau=10^{-3}$.}
\label{fig:illustration5}
\end{figure}
\end{example}

\section{Conclusion}\label{sec:final}

We have given a simple method to generate a bounded invariant for a differential system.
We have seen that the method can be used to show that, for a parameterized 
differential system, the solutions never converge to an equilibrium point, for any value of the parameter varying over a certain domain. The method uses a very general criterion of inclusion of one set in another. This is more general and a priori simpler to show than the classical contraction property which states that any two trajectories converge exponentially towards each other,
and often requires the finding of a Lyapunov function.
We have illustrated the interest of the method on the parameterized
example of a VdP system.
We have also pointed out some limits of the method which cannot account for convergence to a limit cycle for complex systems whereas it is feasible with specialized methods using Fourier analysis. In future work, we plan to to extend our method 
in order to account for such an analysis.

%Ets-il vrai que $\bigcup_{t\in[iT,(i+1)T]}B(t)$ (abrege en $B([iT,(i+1)T])$)est un {\em invariant} (born\'e qui contient toutes les solutions $x(t)$issues de $B_0$, pour $t\in[iT,\infty)$) ? Pas evident que $x(t)\in B[iT,(i+1)T]$ implique $x(t+1)\in B[iT,(i+1)T]$.\\

\newpage
\bibliographystyle{plain}
\bibliography{euler}

\end{document}